\documentclass[a4paper]{llncs}

\usepackage{amsmath}
\usepackage{amssymb}

\usepackage{graphicx}
\usepackage{color}
\usepackage{wrapfig}
\usepackage{fullpage}
\allowdisplaybreaks

\newcommand{\cF}{\mathcal{F}}
\newcommand{\cL}{\mathcal{L}}

\newcommand{\cS}{\mathcal{S}}

\newcommand{\cU}{\mathcal{U}}

\spnewtheorem{myclaim}{Claim}{\itshape}{\upshape}

\title{Efficient Enumeration of the Directed Binary Perfect Phylogenies from Incomplete Data\thanks{Partially supported by Grant-in-Aid for Scientific Research from Ministry
  of Education, Science and Culture, Japan, and
  Japan Society for the Promotion of Science, and by Exploratory Research for Advanced Technology (ERATO) from Japan Science and Technology Agency.
The extended abstract version of this paper appears in 11th International Symposium on Experimental Algorithms (SEA 2012) \cite{thispaper-sea}.
}}

\titlerunning{Enumeration of the Directed Binary Perfect Phylogenies}

\author{%
Masashi Kiyomi\inst{1}
\and
Yoshio Okamoto\inst{2}
\and
Toshiki Saitoh\inst{3}
}

\institute{
School of Information Science,
Japan Advanced Institute of Science and Technology,
Nomi, Japan
\and
Center for Graduate Education Initiative,
Japan Advanced Institute of Science and Technology,
Nomi, Japan
\and
ERATO Minato Discrete Structure Manipulation System Project,
Japan Technology and Science Agency,
Sapporo, Japan
}

\date{\today}

\begin{document}

\maketitle

\begin{abstract}
We study a character-based phylogeny reconstruction problem when 
an incomplete set of data is given.
More specifically, we consider the situation under the 
directed perfect phylogeny 
assumption with binary characters in which for some species the 
states of some characters are missing.
Our main object is to give an efficient algorithm to enumerate (or list)
all perfect phylogenies that can be obtained when the missing entries are
completed.
While a simple branch-and-bound algorithm (B{\&}B) shows a theoretically 
good performance, we propose another approach based on a zero-suppressed binary 
decision diagram (ZDD).
Experimental results on randomly generated data exhibit that 
the ZDD approach outperforms B{\&}B\@.
We also prove that counting the number of phylogenetic trees consistent with 
a given data is \#{P}-complete, thus providing an evidence that an efficient 
random sampling seems hard.
\end{abstract}

\section{Introduction}

One of the most important problems in phylogenetics is reconstruction of 
phylogenetic trees.
In this paper, we focus on the character-based approach.
Namely, each species is described by their characters, and a mutation corresponds to 
a change of characters.
However, in the real-world data not all states of all characters are observable or
reliable, which makes the data incomplete.
Thus, we need a methodology that can cope with such incompleteness.

Following Pe'er et al.\ \cite{DBLP:journals/siamcomp/PeerPSS04},
we work with the \emph{perfect phylogeny} assumption, which means that the set of all nodes
with the same character state induces a connected subtree.
All characters are \emph{binary}, namely take only two values.
Without loss of generality, assume that these two values are encoded by $0$ and $1$.
Then, the phylogeny is \emph{directed} in a sense that
for each character a mutation from $0$ to $1$ is possible only once, but
a mutation from $1$ to $0$ is impossible
(this is also called the Camin--Sokal parsimony \cite{caminsokal}).
We consider the situation where for some species the states of some
characters are unknown.
Under this setting,
Pe'er et al.\ \cite{DBLP:journals/siamcomp/PeerPSS04}
provided a polynomial-time algorithm to reconstruct a 
phylogenetic tree that can be obtained when the unknown states are 
completed, if it exists.

Although their algorithm can find a phylogenetic tree efficiently, it does not 
take the likelihood into account.
This motivates people to look at optimization problems; namely we may introduce an objective
function (or an evaluation function) and try to find a perfect phylogeny that 
maximizes the value of the function.
For example, Gusfield et al.\ \cite{DBLP:conf/cocoon/GusfieldFB07} looked at such an 
optimization problem and formulated it as an integer linear program.
One big issue here is that these optimization problems tend to be NP-hard, and thus
we cannot expect to obtain polynomial-time algorithms.
Therefore, we need some compromise.
If we insist on efficiency, then we need to sacrifice the quality of an obtained
solution.
This approach leads us to approximation algorithms.
If we insist on optimality, then we need to sacrifice the running time.
This approach leads us to exponential-time exact algorithms.
However, techniques in the literature as Gusfield et al.\ \cite{DBLP:conf/cocoon/GusfieldFB07} with these approaches use specific 
structures of the form of objective functions.

\subsection{Our Results}

The focus of this paper is the exact approach.
However, unlike the previous work, 
we aim at \emph{enumeration algorithms}, which give a more flexible framework 
for scientific discovery independent of the form of objective functions.
The use of enumeration algorithms is highlighted in data mining and artificial 
intelligence.
For example, the apriori algorithm by Agrawal and Srikant
\cite{DBLP:conf/vldb/AgrawalS94} enumerates all maximal frequent 
itemsets in a transaction database.
It is not expected that such enumeration algorithms run faster than non-enumeration 
algorithms.  
Therefore, the goal of this paper is to examine
 a possibility and a limitation of enumerative approaches.

One of the difficulties in designing efficient enumeration algorithms is to avoid 
duplication.
Suppose that we are to output an object, and need to check if this object was already
output or not.
If we store all objects that we output so far, then we can check it by going through
them.
However, storing them may take too much space, and going through them may take too much time.
The number of obejcts is typically exponentially large.
Our algorithm cleverly avoids such checks, but still ensures 
exhaustive enumeration without duplication.

It is rather straightforward to give an algorithm with theoretical 
guarantee such as polynomiality.
Namely, a simple branch-and-bound idea gives an algorithm that 
has a running time polynomial in the input size and 
linear in the output size.
Notice that an enumeration algorithm outputs all the objects, and thus 
the running time needs to be at least as high as the number of output objects.
Thus, the linearity in the output size cannot be avoided
in any enumeration algorithms.

However, such a theoretically-guaranteed algorithm does not necessarily 
run fast in practice.
Thus, we propose another algorithm that is based on 
a zero-suppressed binary decision diagram (ZDD).
A ZDD was introduced by Minato \cite{DBLP:conf/dac/Minato93}.
It is a directed graph that has a  
similar structure to a binary decision diagram (BDD).
While a BDD is used to represent a boolean function 
in a compressed way, a ZDD only represents
the satisfying assignments of the function in a compressed way
(a formal definition will be given in Section \ref{sec:zdd}).
Furthermore, we may employ a lot of operations on ZDDs, called the family algebra,
which can be used for efficient filtering and optimization with respect to 
some objective functions.
A book of Knuth \cite{knuth4-1} devotes one section to ZDDs, and gives
numerous applications as exercises.

Although the size of a constructed ZDD is bounded by a polynomial of the number of 
output objects,
we cannot guarantee that the size of a ZDD that is
 created at the intermediate steps in the course 
of our algorithm is bounded.
This means that we cannot guarantee a polynomial-time running time
(in the input size and the output size) for our ZDD algorithm.
However, the crux here is that the size of a constructed ZDD can be 
much smaller than the number of output objects.
We exhibit this phenomenon in two ways.
First, we give an example in which the number of phylogenetic trees  
is exponential in the input size, but the size of the constructed 
ZDD is polynomial in the input size.
Second, we perform experiments on randomly generated data, and the result 
shows that our ZDD algorithm can solve more instances than a
branch-and-bound algorithm.
This suggests that the ZDD approach is quite promising.

Having enumeration algorithms, we can also count the number of 
phylogenetic trees.
In particular, the branch-and-bound algorithm can count them in polynomial time 
in the input size and the output size.
This naturally raises the following question:
Is it possible to count them in polynomial time only in the input size?
Note that since we only compute the number, we do not have to output each 
object one by one, and thus the linearity of the running time in the output 
size could be avoided.
Such a polynomial-time counting algorithm could be combined with 
a branch-and-bound enumeration algorithm to design a random sampling algorithm.
Namely, when we branch, we count the number of outputs in each subinstance
in polynomial time, and choose one subinstance at random according to the 
computed numbers.
For more on the connection of counting and sampling, 
we refer to a book by Sinclair \cite{sinclair}. 

We prove that this is unlikely.
Namely, counting the number of phylogenetic trees for the incomplete directed 
binary perfect phylogeny is \#{P}-complete.
The complexity class \#{P} contains all counting problems in which a counted object 
has a polynomial-time verifiable certificate.
Since no \#{P}-complete problem is known to be solved in polynomial time, 
the \#{P}-completeness suggests the unlikeliness for the problem to be solved in 
polynomial time.

\subsection{Graph Sandwich}

Pe'er et al.\ \cite{DBLP:journals/siamcomp/PeerPSS04} rephrased 
the incomplete directed binary perfect phylogeny problem
as a bipartite graph sandwich problem.
The graph sandwich problem, in general, was introduced by Golumbic
et al.\ \cite{DBLP:journals/jal/GolumbicKS95}.
In the graph sandwich problem, we fix a class $C$ of graphs, and
we are given two graphs $G_1=(V,E_1), G_2=(V,E_2)$ such that
$E_1 \subseteq E_2$.
Then, we are asked to find a graph $G=(V,E)\in C$ such that
$E_1 \subseteq E \subseteq E_2$.
Golumbic et al.\ \cite{DBLP:journals/jal/GolumbicKS95} proved that
even for some
restricted classes of graphs, the problem is NP-complete.
The subsequent results by various researchers also show that 
for a lot of cases the problem is NP-complete, even though the recognition problem 
for those classes can be solved in polynomial time
(we will not include here a long list of literature).
Thus, the result by Pe'er et al.\ \cite{DBLP:journals/siamcomp/PeerPSS04} 
gives a rare example for which the graph sandwich problem can be solved 
in polynomial time.

Recently, the graph sandwich enumeration problem has been studied.
Kijima et al.\ \cite{DBLP:journals/tcs/KijimaKOU10} studied the graph sandwich 
enumeration problem for chordal graphs.
They provided efficient algorithms when $G_1$ or $G_2$ is chordal, where 
``efficient'' means that it runs in polynomial time in the input size and linear 
time in the output size.
Their approach was generalized by Heggernes et al.\ \cite{DBLP:journals/jco/HeggernesMPS11}
to all sandwich-monotone graph classes.
In this respect, this paper gives another example of 
efficient graph sandwich enumeration algorithms.

\subsection{Organization}

In Section \ref{sec:prelim}, 
we introduce the problem more formally.
In Section \ref{sec:zdd}, we provide the algorithm based on 
ZDDs, and give an example in which the compression really works.
In Section \ref{sec:hard}, we prove that the counting version is 
intractable.
Section \ref{sec:experiment} gives experimental results.
We conclude in the final section.

\section{Preliminaries}
\label{sec:prelim}

Due to the pairwise compatibility lemma (see, e.g., \cite{DBLP:reference/algo/Jansson08}), 
we may define our problem in terms of laminars.
We adapt this view throughout the paper.

A sequence $\cS=(S_1,\ldots,S_m)$ of subsets of a finite set $S$ is a \emph{laminar} if for 
every two $i,j\in \{1,\ldots,m\}$ the intersection $S_i \cap S_j$ is either $S_i$, $S_j$, or $\emptyset$.%
\footnote{Usually, a laminar is defined as a family of subsets, but for our purpose it is 
convenient to define as a sequence of subsets.}
In the \emph{incomplete directed binary perfect phylogeny problem} (IDBPP), we are given two sequences
$\cL=(L_1,\ldots,L_m)$, $\cU=(U_1,\ldots,U_m)$ of $m$ subsets of $S$ such that 
$L_i \subseteq U_i \subseteq S$ for all $i\in\{1,\ldots,m\}$, and the question is to determine whether 
there exists a laminar $\cS=(S_1,\ldots,S_m)$ such that $L_i \subseteq S_i \subseteq U_i$ for all
$i\in\{1,\ldots,m\}$.
We call such a laminar a \emph{directed binary perfect phylogeny} for $(S,\cL,\cU)$.
The IDBPP can be solved in polynomial time \cite{DBLP:journals/siamcomp/PeerPSS04}.

Let us briefly describe the correspondence to phylogenetic trees.
The set $S$ represents the set of species, and the indices $1,\ldots,m$ represent
the characters.
Then, $S_i$ represents the set of species that has the character $i$.
The species in $L_i$ are recognized as those we know having the character $i$, and 
the species in $S\setminus U_i$ are recognized as those we know not having $i$.

In this paper, we consider the following variants that take the same input as the IDBPP\@.
In the \emph{counting version} of IDBPP\@, the objective is to output the number of directed binary 
perfect phylogenies.
In the \emph{enumeration version} of IDBPP\@, the objective is to output all the directed binary perfect 
phylogenies.
Note that enumeration should be exhaustive, and also should not output the same object twice or more.

\section{ZDD Approach}
\label{sec:zdd}

\subsection{Introduction to ZDDs}

Let $f\colon \{0,1\}^N \to \{0,1\}$ be an $N$-variate boolean function with 
boolean variables $x_1,\ldots,x_N$.
We assume a linear order on the variables $\{x_1,\ldots,x_N\}$ as 
$x_i$ precedes $x_j$ if and only if $i<j$.
A \emph{binary decision diagram} (BDD) for $f$, denoted by $B(f)$, 
is a vertex-labeled directed graph with the following properties.

\begin{itemize}
\item There is only one vertex with indegree $0$, called the \emph{root} of $B(f)$.
\item There are only two vertices with outdegree $0$, called the \emph{terminals} of $B(f)$.
\item Each vertex of $B(f)$, except for the terminals, is labeled by a variable  from $\{x_1,\ldots,x_N\}$.
\item One terminal is labeled by $0$ (called the \emph{$0$-terminal}), and the other terminal is labeled by $1$
(called the \emph{$1$-terminal}).
\item Each edge of $B(f)$ is labeled by $0$ or $1$.
An edge labeled by $0$ is called a \emph{$0$-edge}, and an edge labeled by $1$ is called a 
\emph{$1$-edge}.
\item Each vertex of $B(f)$, except for the terminals, has exactly one outgoing $0$-edge and
exactly one outgoing $1$-edge.
\item If there is a path from a vertex $v$ to a non-terminal vertex $u$ in $B(f)$, then the 
label of $v$ is smaller than the label of $u$.
\item 
A boolean assignment $\alpha\colon \{x_1,\ldots,x_N\} \to \{0,1\}$ satisfies $f$ (i.e., $f(\alpha(x_1), \ldots, \alpha(x_N))=1$)
if and only if there exists a path $P$ from the root to the $1$-terminal in $B(f)$
that satisfies the following condition:
$\alpha(x_i)=1$ if and only if there exists a vertex $v$ on $P$ labeled by $x_i$ 
such that $P$ traverses the $1$-edge leaving $v$.
\end{itemize}

\begin{wrapfigure}{r}{2.5cm}
\centering
\resizebox{2cm}{!}{\includegraphics{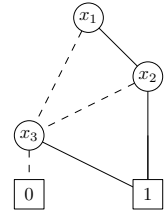}}
\caption{A ZDD for the function $f(x_1,x_2,x_3)=(x_1 \wedge x_2 \wedge \overline{x_3}) \vee (\overline{x_2} \wedge x_3)$.}
\label{fig:zddexample}
\end{wrapfigure}
A BDD for a function $f$ is not unique, and may contain redundant information.
However, the following reduction rules turn a BDD into a smaller equivalent BDD\@.
A \emph{zero-suppressed binary decision diagram} (ZDD) for a function $f$ is a BDD, denoted by $Z(f)$,
for which the reduction rules
cannot be applied.
\begin{enumerate}
\item If the outgoing $1$-edge of a vertex $v$ points to the $0$-terminal
and the outgoing $0$-edge of a vertex $v$ points to a vertex $u$,
then we remove $v$ and its outgoing edges, and reconnect the incoming edges to $v$ to the vertex $u$.
\item If two vertices $v,v'$ have the same label $x_i$, their outgoing $1$-edges point to the same vertex $u_1$, and
their outgoing $0$-edges point to the same vertex $u_0$, then replace $v,v'$ with a single vertex $w$ with label
$x_i$.  The incoming edges to $w$ are those to $v,v'$, the outgoing $1$-edge from $w$ points to $u_1$, 
and the outgoing $0$-edge from $w$ points to $u_0$.
\end{enumerate}
\figurename\ \ref{fig:zddexample} shows an example of a ZDD\@.
The edges are assumed to be directed downward.
A dashed line represents a $0$-edge, and 
a solid line represents a $1$-edge.

The \emph{size} of a ZDD $Z(f)$ is defined as the number of vertices, and denoted by $|Z(f)|$.
It is easy to observe that the size of ZDD $Z(f)$ is $O(NA)$ where $A$ is the number of
satisfying assignments of $f$.
However, this is merely an upper bound, and in practice the size can be much smaller.
Thus, a ZDD for $f$ gives a compressed representation of the family of all satisfying assignments of
$f$.
Especially, if we have a family $\cF$ of subsets of a finite set $S$
and consider a boolean function $f\colon \{0,1\}^S \to \{0,1\}$
such that $f(x) = 1$ if and only if $\{e\in S \mid x_e = 1\} \in \cF$, then
a ZDD for $f$ compactly encodes the family $\cF$.

There are a family of operations that can be performed on ZDDs.
Here, we list those which we use in our algorithm.
Let $f,f'\colon \{0,1\}^{N} \to \{0,1\}$ be boolean functions with variables
$x_1,\ldots,x_N$, and ZDDs $Z(f),Z(f')$  be given.
Then, a ZDD $Z(f \vee f')$ of the disjunction (logical OR) can be obtained
in $O(|Z(f)||Z(f')|)$ time.
Let $f^{[x_i=0]}\colon \{0,1\}^{N-1}\to \{0,1\}$ be a boolean function with 
variables $x_1,\ldots,x_{i-1},x_{i+1},\ldots,x_N$ obtained from $f$ by
$f^{[x_i=0]}(x_1,\ldots,x_{i-1},x_{i+1},\ldots,x_N)=f(x_1,\ldots,x_{i-1},0,x_{i+1},\ldots,x_N)$.
Then, a ZDD $Z(f^{[x_i=0]})$ can be found in $O(|Z(f)|)$ time.
Similarly, we may define $f^{[x_i=1]}$, and a ZDD $Z(f^{[x_i=1]})$ can be found in $O(|Z(f)|)$ time.

\subsection{ZDD-Based Enumeration Algorithm}

We introduce a boolean variable $x_{i,e}$ for each pair $(i,e)$
of an index $i\in\{1,\ldots,m\}$ and an element $e\in S$.
Then, we consider the conjunction (logical AND) of the following conditions, 
which gives rise to a boolean function $f\colon \{0,1\}^{\{1,\ldots,m\}\times S} \to \{0,1\}$.
\begin{enumerate}
\item For every $i\in\{1,\ldots,m\}$, if $e\in L_i$, then
$x_{i,e}=1$.
\item For every $i\in\{1,\ldots,m\}$, if $e\in S \setminus U_i$, then 
$x_{i,e}=0$.
\item For every distinct $i,j\in \{1,\ldots,m\}$, exactly one of the 
following three is satisfied.
\begin{enumerate}
\item For all $e\in S$, if $x_{i,e}=1$, then $x_{j,e}=1$.
\item For all $e\in S$, if $x_{i,e}=0$, then $x_{j,e}=0$.
\item For all $e\in S$, if $x_{i,e}=1$, then $x_{j,e}=0$.
\end{enumerate}
\end{enumerate}
We can easily see that if we set 
$S_i = \{e \in S \mid x_{i,e}=1\}$ for every $i\in\{1,\ldots,m\}$,
then $\cS=(S_1,\ldots,S_m)$ is a directed binary perfect phylogeny 
for $(S,\cL,\cU)$.
Namely, 
the condition 1 translates to $L_i \subseteq S_i$;
the condition 2 translates to $S_i \subseteq U_i$;
the condition 3(a) translates to $S_i \cap S_j = S_i$;
the condition 3(b) translates to $S_i \cap S_j = S_j$;
the condition 3(c) translates to $S_i \cap S_j = \emptyset$.

These conditions naturally induce the following algorithm.

\begin{description}
\item[Algorithm:] $\mathsf{ZDD}(S,\cL,\cU)$
\item[Precondition: ] $S$ is a finite set, $\cL=(L_1,\ldots,L_m)$, $\cU=(U_1,\ldots,U_m)$, each
member of $\cL$ and $\cU$ is a subset of $S$, and 
$L_i \subseteq U_i$ for every $i\in\{1,\ldots,m\}$.
\item[Postcondition:  ] Output a ZDD $Z(f)$ for the boolean function $f$ over
the variables $\{x_{i,e} \mid i\in\{1,\ldots,m\}, e\in S\}$ defined above, which 
encodes all the directed binary perfect phylogenies for $(S,\cL,\cU)$.
\item[Step 0: ] Let $g=\vec{1}$ be the constant-one function.
Construct a ZDD $Z(g)$.
\item[Step 1: ] For each $i\in \{1,\ldots,m\}$ and each $e\in S$, if 
$e\in L_i$, then construct $Z(g^{[x_{i,e}=1]})$ from $Z(g)$ and reset
$g:=g^{[x_{i,e}=1]}$.
\item[Step 2: ] For each $i\in \{1,\ldots,m\}$ and each $e\in S$, if
$e\in S\setminus U_i$, then construct $Z(g^{[x_{i,e}=0]})$ from $Z(g)$ and 
reset $g:=g^{[x_{i,e}=0]}$.
\item[Step 3:] For each distinct $i,j\in \{1,\ldots,n\}$ and each $e\in S$, we perform the following.
\item[Step 3-a:] Let $g_1 := g^{[x_{i,e}=1, x_{j,e}=1]} \vee g^{[x_{i,e}=0]}$. 
Construct $Z(g_1)$ from $Z(g)$.
\item[Step 3-b:] Let $g_2 := g^{[x_{i,e}=0, x_{j,e}=0]} \vee g^{[x_{i,e}=1]}$. 
Construct $Z(g_2)$ from $Z(g)$.
\item[Step 3-c:] Let $g_3 := g^{[x_{i,e}=1, x_{j,e}=0]} \vee g^{[x_{i,e}=0]}$. 
Construct $Z(g_3)$ from $Z(g)$.
\item[Step 3-d:] Construct $Z(g_1 \vee g_2 \vee g_3)$ from $Z(g_1), Z(g_2), Z(g_3)$, and
reset $g:=g_1 \vee g_2 \vee g_3$.
\item[Step 4:] Output $Z(g)$ and halt.
\end{description}

Although the output size $|Z(f)|$ is bounded by $O(mnh)$ where $n=|S|$ and $h$ is the number of 
directed binary perfect phylogenies for $(S,\cL,\cU)$, we cannot guarantee that ZDDs that appear in 
the course of execution have such a bounded size.
Thus, the algorithm could be quite slow
 or could stop due to memory shortage.

\subsection{Example with Huge Compression}

We exhibit an example for which the size of a ZDD is exponentially smaller than 
the number of directed binary perfect phylogenies.
While the example is artificial, this indicates a possibility that our ZDD-based algorithm 
outperforms the branch-and-bound algorithm.

Consider the following example.
Let $S=\{(i,j) \mid i\in\{1,\ldots,n\}, j\in\{0,1,\ldots,k\}\}$.
Then $|S|=(k+1)n$.
For each $i\in \{1,\ldots,n\}$, let 
$L_i = \{(i,0)\}$ and $U_i = \{(i,0),(i,1),\ldots,(i,k)\}$.
As before, let $\cL=(L_1,\ldots,L_n)$ and $\cU=(U_1,\ldots,U_n)$.
\begin{proposition}
The number of directed binary perfect phylogenies for $(S,\cL,\cU)$ is 
$2^{kn}$.
\end{proposition}

\begin{proof}
For two distinct $i,j\in \{1,\ldots,n\}$, it holds that $U_i \cap U_j =\emptyset$.
Therefore, for any subsets $S_i \subseteq U_i \setminus L_i$ and $S_j \subseteq U_j \setminus L_j$,
it holds that $S_i \cap S_j = \emptyset$.
This means that a directed binary perfect phylogeny for $(S,\cL,\cU)$ can be formed by 
choosing an arbitrary subset of $U_i \setminus L_i$ for each $i\in\{1,\ldots,n\}$.
Since $|U_i \setminus L_i|=k$, the number of subsets of $U_i \setminus L_i$ is $2^k$, and 
thus the number of directed binary perfect phylogenies is $(2^k)^n = 2^{kn}$.
\qed
\end{proof}

\begin{proposition}
The size of a ZDD constructed by $\mathsf{ZDD}(S,\cL,\cU)$ is $O(kn)$.
\end{proposition}

\begin{proof}
\figurename\ \ref{fig:linear} shows a constructed ZDD\@.
Note that an ordering of variables is not relevant.
No matter which ordering we impose on the variables,
we obtain an isomorphic ZDD\@.
\qed
\end{proof}

\begin{figure}[tp]
\centering
\resizebox{!}{.6\textheight}{\includegraphics{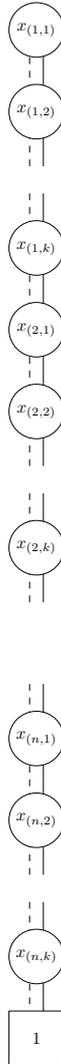}}
\caption{An example for which the number of directed binary 
perfect phylogenies is exponential, but the size of a ZDD is linear.}
\label{fig:linear}
\end{figure}

\section{Hardness of Counting}
\label{sec:hard}

As we explained in the introduction, an efficient counting algorithm can be 
used to efficient sampling of combinatorial objects.
In this section, we prove that it is unlikely that such an algorithm exists for 
the IDBPP by showing that the counting version is \#{P}-complete.

\begin{theorem}
\label{thm:sharpP}
The counting version of the IDBPP is
\#{P}-complete.
\end{theorem}

\begin{proof}
We reduce the problem of counting the number of matchings in a 
(simple) bipartite graph,
which is known to be \#{P}-complete \cite{DBLP:journals/siamcomp/Valiant79}.

Let $G=(V,E)$ be a (simple) bipartite graph with a bipartition $V=A\cup B$ of the vertex set.
For each vertex $v\in V$, we set up an element $s_v$, and
let $S=(s_v \mid v\in V)$.
Then, for each edge $e = \{a,b\} \in E$, where $a\in A$ and $b\in B$, 
let $L_e = \{s_a\}$ and $U_e=\{s_a, s_b\}$.
Then, we set up 
$\cL= (L_e \mid e\in E)$ and $\cU=(U_e \mid e\in E)$.
Note that for each $e\in E$, it holds that $L_e \subseteq U_e$.
Thus, $S$, $\cL$, and $\cU$ form an instance of the IDBPP\@.

Let $\cS=(S_e \mid e\in E)$ be a directed binary perfect phylogeny for
$(S,\cL,\cU)$.
Then, $S_e$ is either $L_e$ or $U_e$ for every $e\in E$, 
since $|L_e|=1$, $|U_e|=2$, and $L_e \subseteq S_e \subseteq U_e$.
\begin{myclaim}
\label{myclaim:A}
Let $\cS=(S_e \mid e\in E)$ be a directed binary perfect phylogeny for
$(S,\cL,\cU)$.
Then, the set $M=\{e \in E \mid S_e = U_e\}$ is a matching of $G$.
\end{myclaim}
\begin{proof}[of Claim \ref{myclaim:A}]
Suppose not.
Then, there exist two distinct edges $e,e' \in M$ that share an endpoint, say $v$.
This means that $s_v \in S_e \cap S_{e'}$.
Since $\cS$ is a laminar on $S$, it must hold that $S_e \subseteq S_{e'}$ or
$S_{e'} \subseteq S_e$.
Since $|S_e|=2=|S_{e'}|$, it follows that $S_e = S_{e'}$.
Then, $e = e'$ since $G$ is a simple graph.
This contradicts the assumption that $e$ and $e'$ are distinct edges.
\qed
\end{proof}
The following claim shows the converse.
\begin{myclaim}
\label{myclaim:B}
Let $M\subseteq E$ be a matching of $G$.
Then, the following $\cS=(S_e \mid e\in E)$ is a directed binary perfect phylogeny
for $(S,\cL,\cU)$:
$S_e = U_e$ if $e\in M$, and 
$S_e = L_e$ otherwise.
\end{myclaim}
\begin{proof}[of Claim \ref{myclaim:B}]
It suffices to prove that the constructed sequence $\cS$ is a laminar.
Consider two sets $S_e, S_{e'}$ for two distinct $e,e'\in E$.
We have three cases.
Let $e=\{a,b\}$ and $e'=\{a',b'\}$, where $a,a' \in A$ and $b,b'\in B$.
\begin{enumerate}
\item Assume that $e\in M$ and $e'\in M$. 
Then, $\{a,b\} \cap \{a',b'\} = \emptyset$, and therefore $S_e \cap S_{e'} = \emptyset$.
\item Assume that $e\in M$ and $e'\not\in M$.
If $a\neq a'$, then $S_e \cap S_{e'} = \emptyset$.
If $a=a'$, then $S_{e'} = \{s_{a'}\} \subseteq \{s_a,s_b\} = S_e$.
Therefore, $S_e \cap S_{e'} = S_{e'}$.
\item Assume that $e\not\in M$ and $e'\not\in M$.
If $a\neq a'$, then $S_e \cap S_{e'} = \emptyset$.
If $a=a'$, then $S_e = S_{e'}$.
\qed
\end{enumerate}
\end{proof}
By the claims above, the number of matchings in $G$ is equal to the number of
directed binary perfect phylogenies for $\cL$ and $\cU$.
Note that the reduction runs in polynomial time.
\qed
\end{proof}

\section{Experiments}
\label{sec:experiment}

\subsection{Data}

We have used the program \textsf{ms} by Hudson \cite{DBLP:journals/bioinformatics/Hudson02} to generate 
a random data set without incompleteness that admits a directed binary perfect phylogeny $\cS=(S_1,\ldots,S_m)$.
Then, we have constructed $L_i$ from $S_i$ by removing each element of $S_i$ 
independently with probability $p$, and 
constructed $U_i$ from $S_i$ by adding each element of $S\setminus S_i$
independently with probability $p$.

We have created $100$ instances independently at random
for each triple of values $(m,n,p) \in \{50,100\} \times \{50,100\} \times \{0.1, 0.2, 0.3, 0.4, 0.5\}$.

\subsection{Implementation and Experiment Environment}

We have implemented the algorithm $\mathsf{ZDD}$ described 
in Section \ref{sec:zdd} and 
another algorithm based on the branch-and-bound idea, which we call 
$\mathsf{B{\&}B}$.
The details of $\mathsf{B{\&}B}$ is deferred to Appendix \ref{sec:bb}.
We have implemented both algorithms in C++\@.
For the implementation of $\mathsf{B{\&}B}$, 
Step 1 uses the deterministic version of Algorithm A in the paper by Pe'er et al.\ 
\cite[p.\ 598]{DBLP:journals/siamcomp/PeerPSS04}, but we have simplified it to gain a 
practical performance.
For example, a set is represented by an integer in such a way that each element of the set
corresponds to a bit in the integer.
For $(n,m)=(50,50)$ we used a 64-bit \verb|unsigned long|, and for other cases we used
two \verb|unsigned long|s.
This enables us to perform each set-theoretic operation 
efficiently by one or two bit operations.
Further, we only count the number of directed binary perfect phylogenies, not 
outputting all of them, to avoid an inessential computation in time measurement.

For the implementation of $\mathsf{ZDD}$, we have used the library BDD+ developed by 
Minato.
Among the variables in $\{x_{i,e} \mid i\in\{1,\ldots,m\}, e\in S\}$, 
those meeting 
the condition 2 were removed beforehand, since the outgoing $1$-edge should point to 
the $0$-terminal.
Furthermore, the variables meeting the condition 1 have been put at the tail of 
the linear order on all variables.
Then, such a variable appears only once as a label of a vertex, since
the outgoing $0$-edge should point to the $0$-terminal.
These have been implemented by combining Steps 0--2 in $\mathsf{ZDD}$\@.
This also affects Step 3: some variables can be further removed, or further 
put at the tail of the linear order.
We have tried to find a complete linear order 
so that the size of the constructed ZDD could be small.
To this end, we have introduced two heuristic methods.
The first one has put the variables in the same $S_i$ as closely as possible.
Since these variables possess heavier dependency, if we would 
put them far, then the ZDD would need to store such dependency at various
locations.
The second one has put the variables in $S_i$ and $S_j$ right 
in front of what were put at the tail, and 
the operations on them corresponding to the condition 3 have been performed later 
in the execution of the algorithm, if $S_i$ and $S_j$ meet more than one
case in the condition 3.

All programs have run on the machine with the following specification;
OS: SUSE Linux Enterprise Server 10 (x86\_64);
CPU: Quad-Core AMD Opteron(tm) Processor 8393 SE
({\#}CPUs 16, {\#}Processors 32, Clock Freq.\ 3092MHz);
Memory: 512GB\@.

\subsection{The Number of Solved Instances}

We have counted the number of instances that were solved by our implementation within two minutes
for $p=0.1, 0.2$.
Here, ``solved'' means that the algorithm successfully halts.
\tablename\ \ref{tab:solvedZDD} shows the result.
As we can see from the table, 
$\mathsf{B{\&}B}$ was not able to solve most of the instances, even if they are small.
On the other hand, $\mathsf{ZDD}$ was able to solve almost all instances when $p=0.1$.
However, when $p=0.2$, the number of solved instances rapidly decreases.

\figurename\ \ref{fig:hist120zdd} shows the accumulated number of solved instances by $\mathsf{ZDD}$.
Note that the horizontal axis is in log-scale.
For $(m,n,p) = (50,50,0.1)$, $\mathsf{ZDD}$ solved each of 
the 99 instances within one second.
For $(m,n,p) = (50,100,0.1)$, it solved each of the 99 instances within five seconds.
This shows high effectiveness of the algorithm $\mathsf{ZDD}$\@.

\begin{table}[bt]
\centering
\caption{The number of solved instances by $\mathsf{B{\&}B}$ and $\mathsf{ZDD}$ out of 100 for each case.}
\begin{tabular}{|r|c|c|c|c||c|c|c|c|}
\hline
&
\multicolumn{4}{|c||}{$\mathsf{B{\&}B}$}
&
\multicolumn{4}{|c|}{$\mathsf{ZDD}$}\\
\cline{2-9}
$(m,n)$
&
$(50,50)$ & $(50,100)$ & $(100,50)$ & $(100,100)$ &
$(50,50)$ & $(50,100)$ & $(100,50)$ & $(100,100)$ \\
\hline
$p=0.1$ & 52        & 17         & 0          & 0 
&
99        & 99         & 93         & 90 \\
$p=0.2$ & 0         & 0          & 0          & 0 
&
57        & 33         & 6          & 4 \\
\hline
\end{tabular}
\label{tab:solvedZDD}
\end{table}

\begin{figure}[t]
\centering
\resizebox{0.99\textwidth}{!}{\rotatebox{-90}{\includegraphics{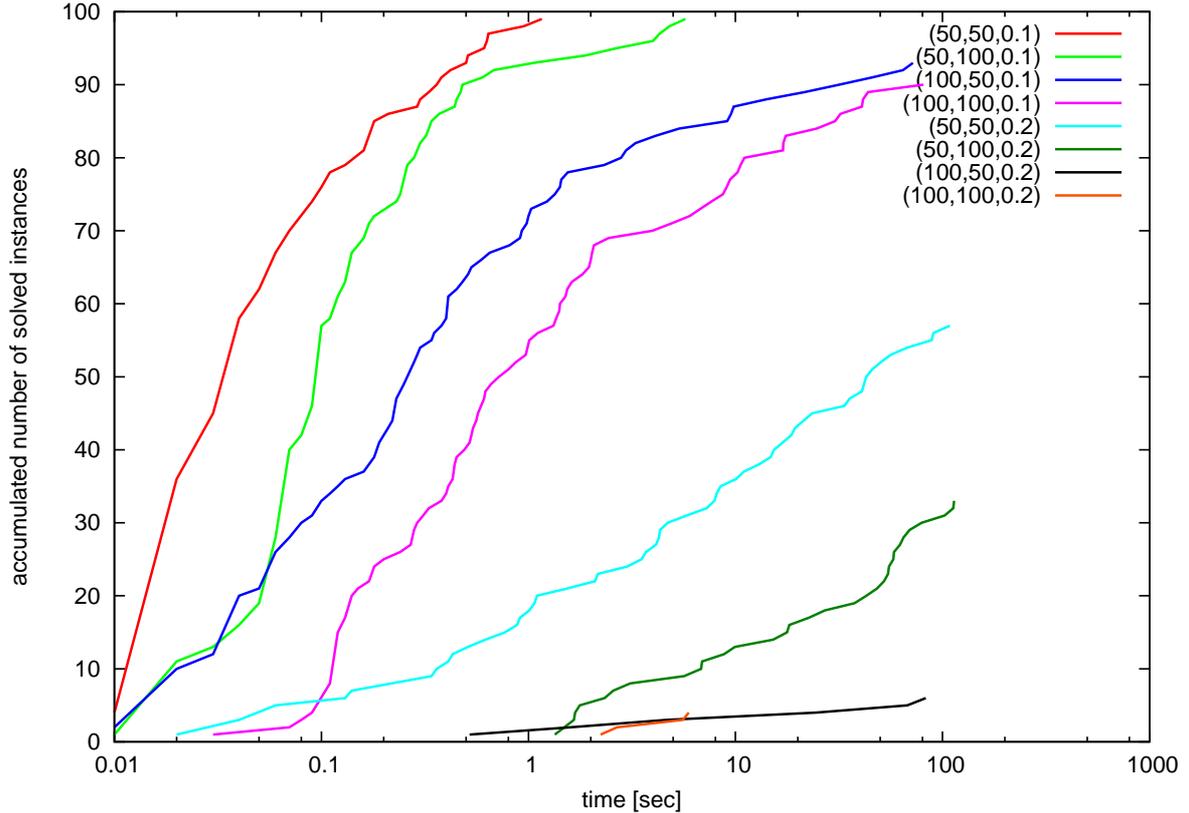}}}
\caption{The accumulated number of instances solved by $\mathsf{ZDD}$ for each case.}
\label{fig:hist120zdd}
\end{figure}

\begin{figure}[t]
\centering
\centering
\resizebox{0.99\textwidth}{!}{\rotatebox{-90}{\includegraphics{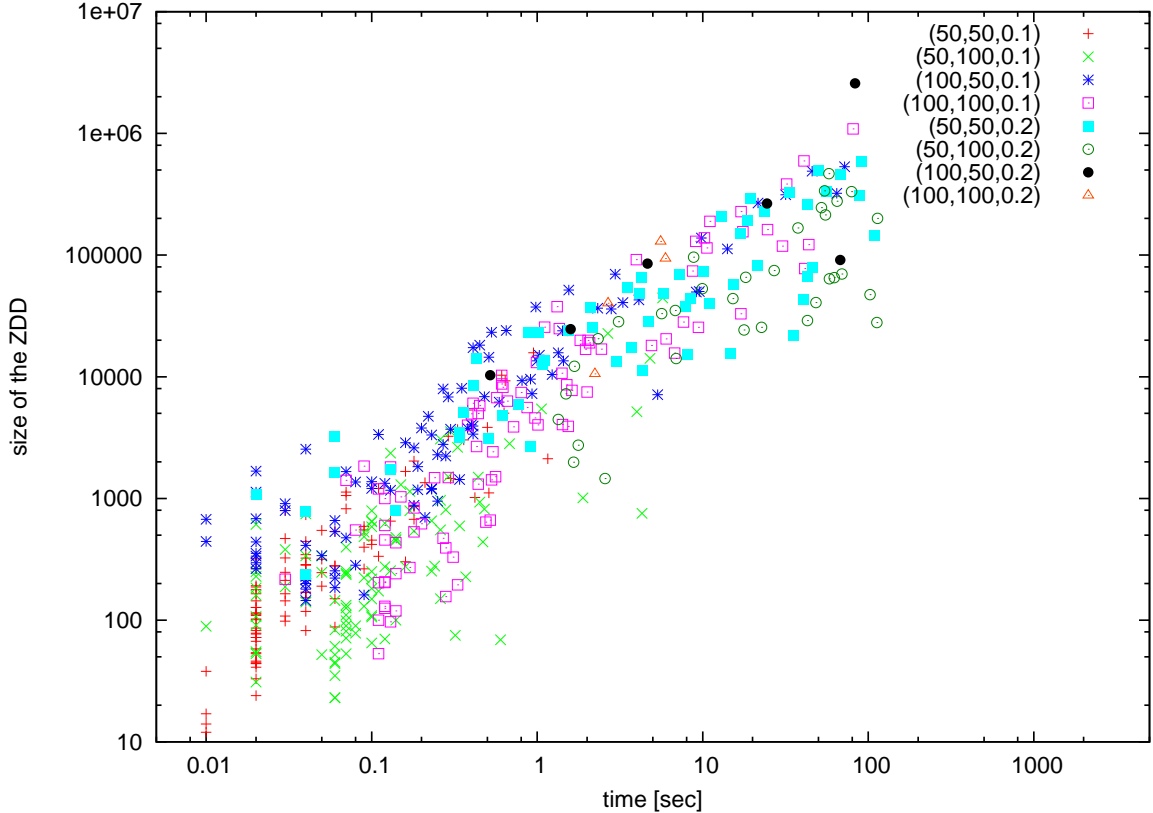}}}
\caption{The size of ZDDs and the running time of $\mathsf{ZDD}$\@.}
\label{fig:timeVSsizeZDD}
\end{figure}

\begin{figure}[t]
\centering
\resizebox{0.99\textwidth}{!}{\rotatebox{-90}{\includegraphics{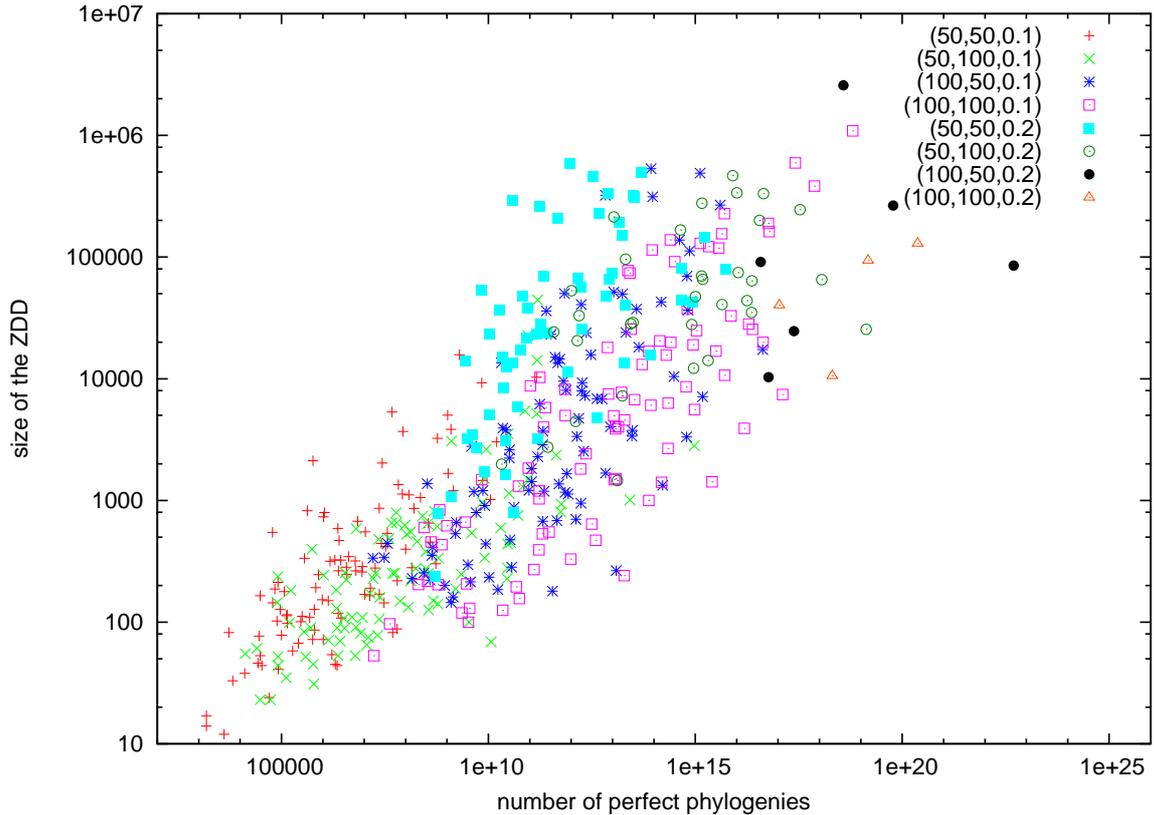}}}
\caption{The number of perfect phylogenies and the size of ZDDs.}
\label{fig:numVSsizeZDD}
\end{figure}

\begin{figure}[t]
\centering
\rotatebox{-90}{\resizebox{!}{0.99\textwidth}{\includegraphics{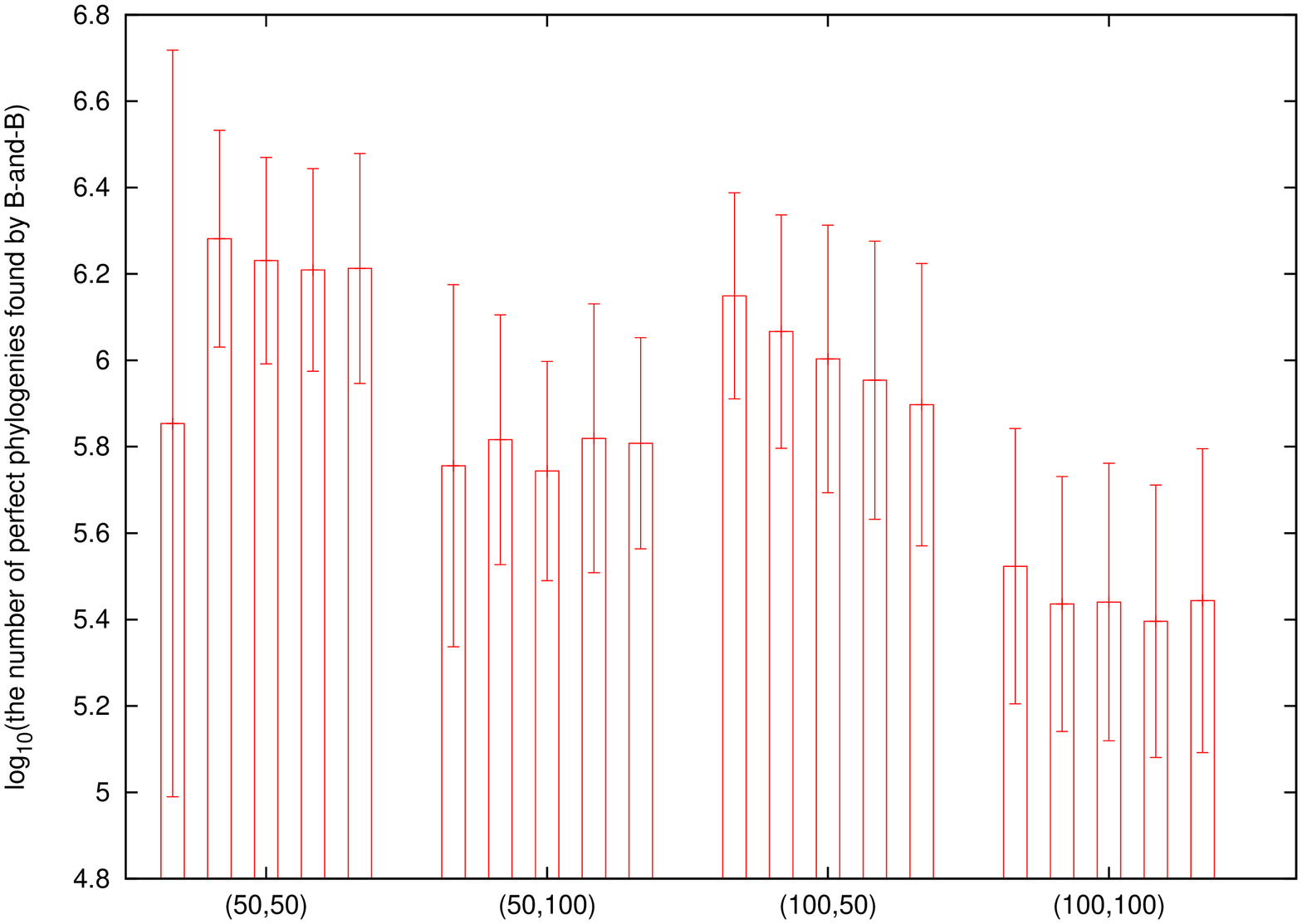}}}
\caption{The number of directed binary perfect phylogenies found by $\mathsf{B{\&}B}$ for each case.}
\label{fig:barBB}
\end{figure}

\subsection{The Running Time of $\mathsf{ZDD}$ and  the Size of ZDDs.}

\figurename\ \ref{fig:timeVSsizeZDD} shows a scatter plot in which each point
represents an instance solved by $\mathsf{ZDD}$ for $p=0.1,0.2$
with the running time (the horizontal coordinate) 
and the size of the ZDD constructed by $\mathsf{ZDD}$ (the vertical coordinate).
Note that this is a log-log plot.
We can see a tendency that the algorithm spends more time for instances with 
larger ZDDs.
A simple $\ell_2$-regression reveals that the spent time is dependent on the size almost linearly.

\subsection{The Number of Perfect Phylogenies and the Size of ZDDs.}

\figurename\ \ref{fig:numVSsizeZDD} shows a log-log scatter plot in which each point
represents an instance solved by $\mathsf{ZDD}$ for $p=0.1,0.2$
with the number of perfect phylogenies (the horizontal coordinate) 
and the size of the ZDD constructed by $\mathsf{ZDD}$ (the vertical coordinate).
The plot exhibits high compression rate of ZDDs.
If we define the \emph{logarithmic compression ratio} of ZDD by the 
logarithm (with base 10) of the size of ZDD divided by the number of 
perfect phylogenies, then \tablename\ \ref{tab:compression} presents the means 
and the standard deviations of the logarithmic compression ratio of 
the instances solved by $\mathsf{ZDD}$ categorized by the choice of parameters.
It shows the high-rate compression by ZDDs, and for larger values 
of parameters the compression ratios get larger.
Among the solved instances, the logarithmic compression ratios range
from $-17.77$ to $-1.82$.
Namely, for the most extreme case, the size of ZDD is approximately 
$10^{17.77}$ times smaller than the number of perfect phylogenies.

\begin{table}[t]
\centering
\caption{The means and the standard deviations of logarithmic compression ratios.}
\begin{tabular}{|c|r|r|r|r|r|r|r|r|}
\hline
$p$     & \multicolumn{4}{|c|}{$0.1$} & \multicolumn{4}{|c|}{$0.2$} \\
\cline{2-9}
$(m,n)$ & 
\multicolumn{1}{|c|}{$(50,50)$} &
\multicolumn{1}{|c|}{$(50,100)$} &
\multicolumn{1}{|c|}{$(100,50)$} &
\multicolumn{1}{|c|}{$(100,100)$} &
\multicolumn{1}{|c|}{$(50,50)$} &
\multicolumn{1}{|c|}{$(50,100)$} &
\multicolumn{1}{|c|}{$(100,50)$} &
\multicolumn{1}{|c|}{$(100,100)$} \\
\hline
mean & 
$-4.13$ & $-7.25$ & $-5.62$ & $-10.00$
& $-8.06$ & $-13.61$ & $-9.24$ & $-14.04$ \\
\hline
standard deviation & 
$1.22$ & $1.35$ & $1.74$ & $1.79$
& $1.48$ & $2.04$ & $1.86$ & $1.02$ \\
\hline
\end{tabular}
\label{tab:compression}
\end{table}

\subsection{The Number of Solutions Found by \textsf{B{\&}B}}

Unlike $\mathsf{ZDD}$, the algorithm $\mathsf{B{\&}B}$ can output some directed binary perfect phylogenies
even if the execution is interrupted.
\figurename\ \ref{fig:barBB} shows the averages of the logarithm
of the numbers of directed binary perfect phylogenies 
(together with standard deviations)
found by $\mathsf{B{\&}B}$ within two minutes for each case: 
Four groups correspond to $(m,n)=(50,50), (50,100), (100,50), (100,100)$
from left to right, and in each group there are five bars corresponding to $p=0.1,0.2,0.3,0.4,0.5$ 
from left to right.
When $(m,n,p) = (50,50,0.1)$, the standard deviation is high since about a half of the instances
were solved within two minutes.
Even for the seemingly difficult case $(m,n,p)=(100,100,0.5)$, $\mathsf{B{\&}B}$ was able to find around
$10^{5.4}$ perfect phylogenies.
This suggests that $\mathsf{B{\&}B}$ can be useful even if $\mathsf{ZDD}$ does not finish the computation.

\begin{figure}[t]
\centering
\rotatebox{-90}{\resizebox{!}{0.99\textwidth}{\includegraphics{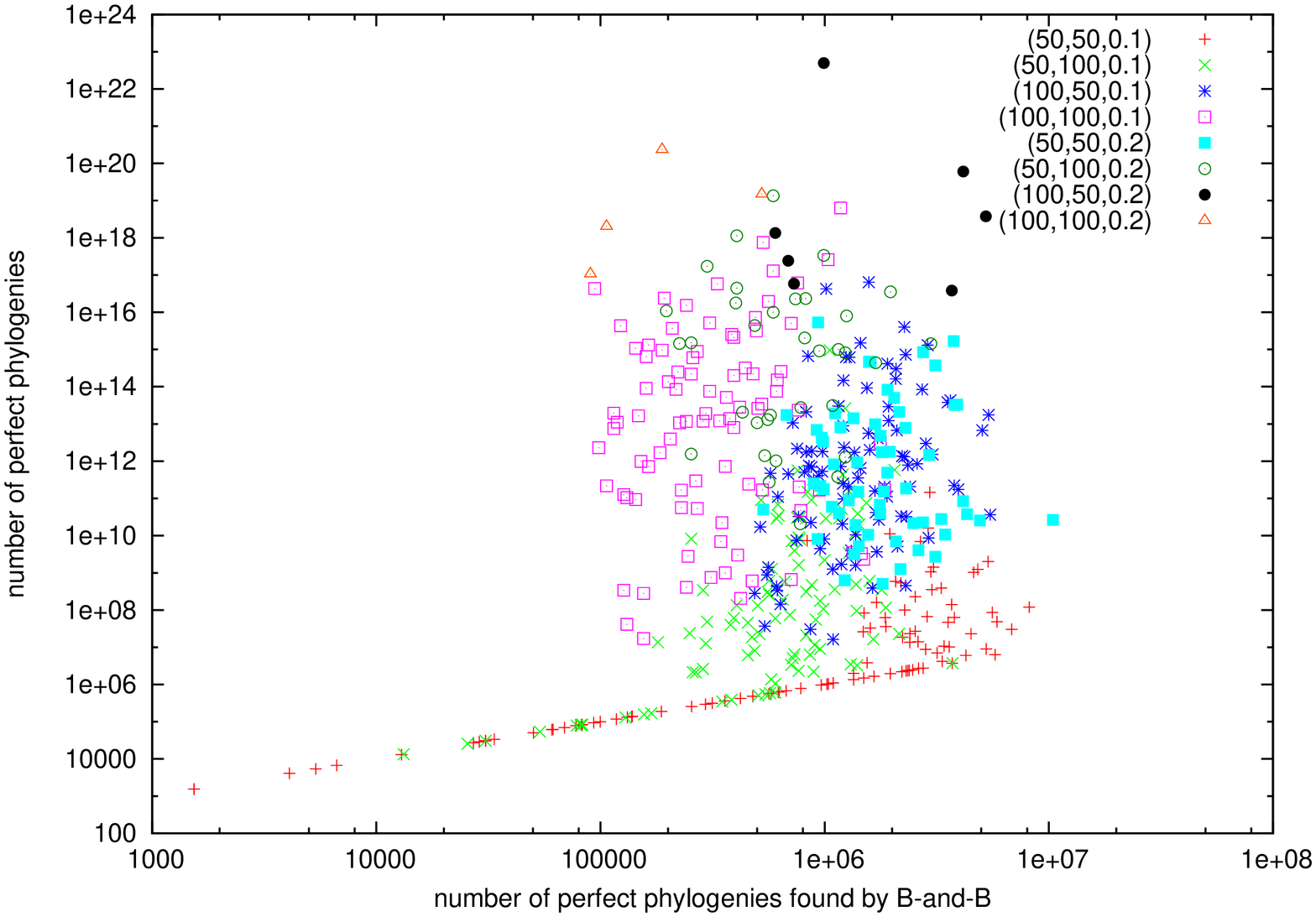}}}
\caption{The number of directed binary perfect phylogenies found by $\mathsf{B{\&}B}$ for each case.}
\label{fig:scatterNumBB}
\end{figure}

\begin{figure}[t]
\centering
\resizebox{0.99\textwidth}{!}{\rotatebox{-90}{\includegraphics{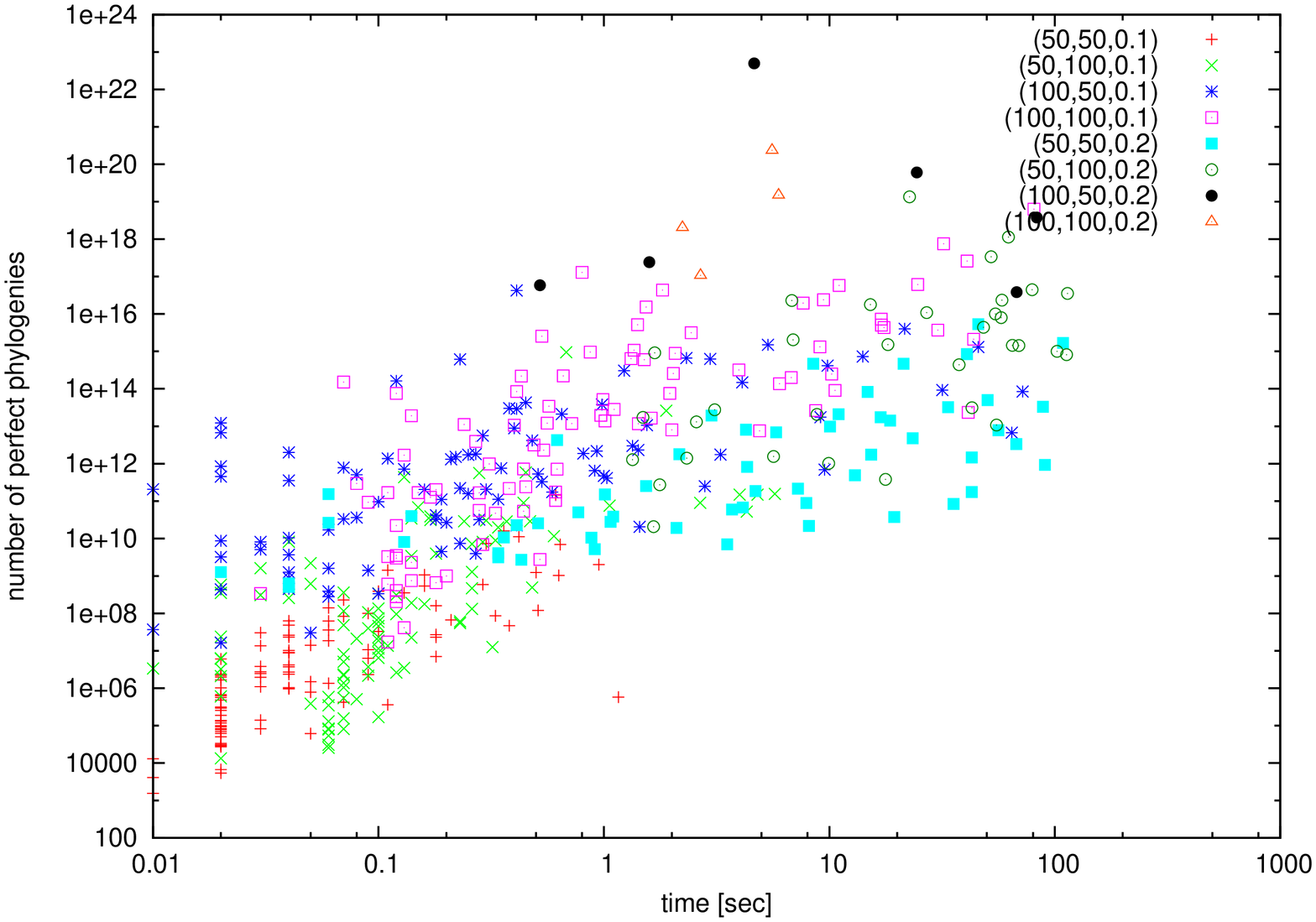}}}
\caption{The number of directed binary perfect phylogenies in the instances solved by $\mathsf{ZDD}$ for each case.}
\label{fig:scatter120zdd}
\end{figure}

\subsection{The Number of Solutions Found by $\mathsf{ZDD}$ and $\mathsf{B{\&}B}$\@.}

\figurename\ \ref{fig:scatterNumBB} is a scatter plot in which each point
represents an instance solved by $\mathsf{ZDD}$ with the number of directed binary 
perfect phylogenies found by $\mathsf{B{\&}B}$ within two minutes (the horizontal coordinate) 
and the number of directed binary perfect phylogenies in the instance (the vertical coordinate).
This shows the percentage of the directed binary perfect phylogenies that were found by 
$\mathsf{B{\&}B}$\@.
Since this is a log-log plot, we can see that this percentage is quite low.
There is one instance for $(m,n,p)=(100,50,0.2)$ with 
49,614,003,829,608,756,019,200 perfect phylogenies 
for which $\mathsf{B{\&}B}$ could only find 991,232.
Thus the percentage is around $10^{-17}$ \%\@.
This really shows the power of ZDDs\@.

\subsection{Running Time of $\mathsf{ZDD}$ and the Number of Solutions}

\figurename\ \ref{fig:scatter120zdd} shows a scatter plot in which each point
represents an instance solved by $\mathsf{ZDD}$ for $p=0.1,0.2$
with the running time (the horizontal coordinate) 
and the number of directed binary perfect phylogenies in the instance (the vertical coordinate).
Note that this is a log-log plot.
There is a weak tendency that the algorithm spends more time for instances with more 
directed binary perfect phylogenies.
We can see that the algorithm is able to solve an instance with more than $10^{17}$ perfect 
phylogenies within one second.

\subsection{The Size of ZDDs During the Execution of $\mathsf{ZDD}$\@.}

\figurename\ \ref{fig:intermediate} traces the size of ZDDs which are created as 
intermediate results during the execution of (the original version of) the algorithm $\mathsf{ZDD}$\@.
In the plot, there are two curves, each of which corresponds to a different instance
for $(m,n,p)=(50,50,0.2)$.
We have measured the size after each execution of Step 3 in the algorithm.
Step 3 is iterated by the number of pairs of distinct integers in $\{1,\ldots,n\}$, which is
$\binom{50}{2}=1,225$.
Therefore, the horizontal coordinates in the plot range from $0$ to $1,224$, and 
the $i$-th iteration gives a point at $i{-}1$ in the horizontal coordinate.
The vertical coordinate corresponds to the size of the ZDD\@.
Notice that this is a semi-log plot.

For the red instance, the algorithm (with heuristic improvements) spent $2.16$ seconds to solve, and 
for the green instance, it spent $108.71$ seconds to solve.
In this sense, the green one is a harder instance than the red one.
As we can see from the figure, the size of ZDDs are changing over time
non-monotonously.
For the red instance, the size of the final result is $25,414$, while 
the maximum size during the execution is $26,174$;
the ratio is $1.03$.
On the other hand, for the green instance, 
the size of the final result is $144,100$, while 
the maximum size during the execution is $271,037$;
the ratio is $1.88$.

\begin{figure}[t]
\centering
\resizebox{0.99\textwidth}{!}{\rotatebox{-90}{\includegraphics{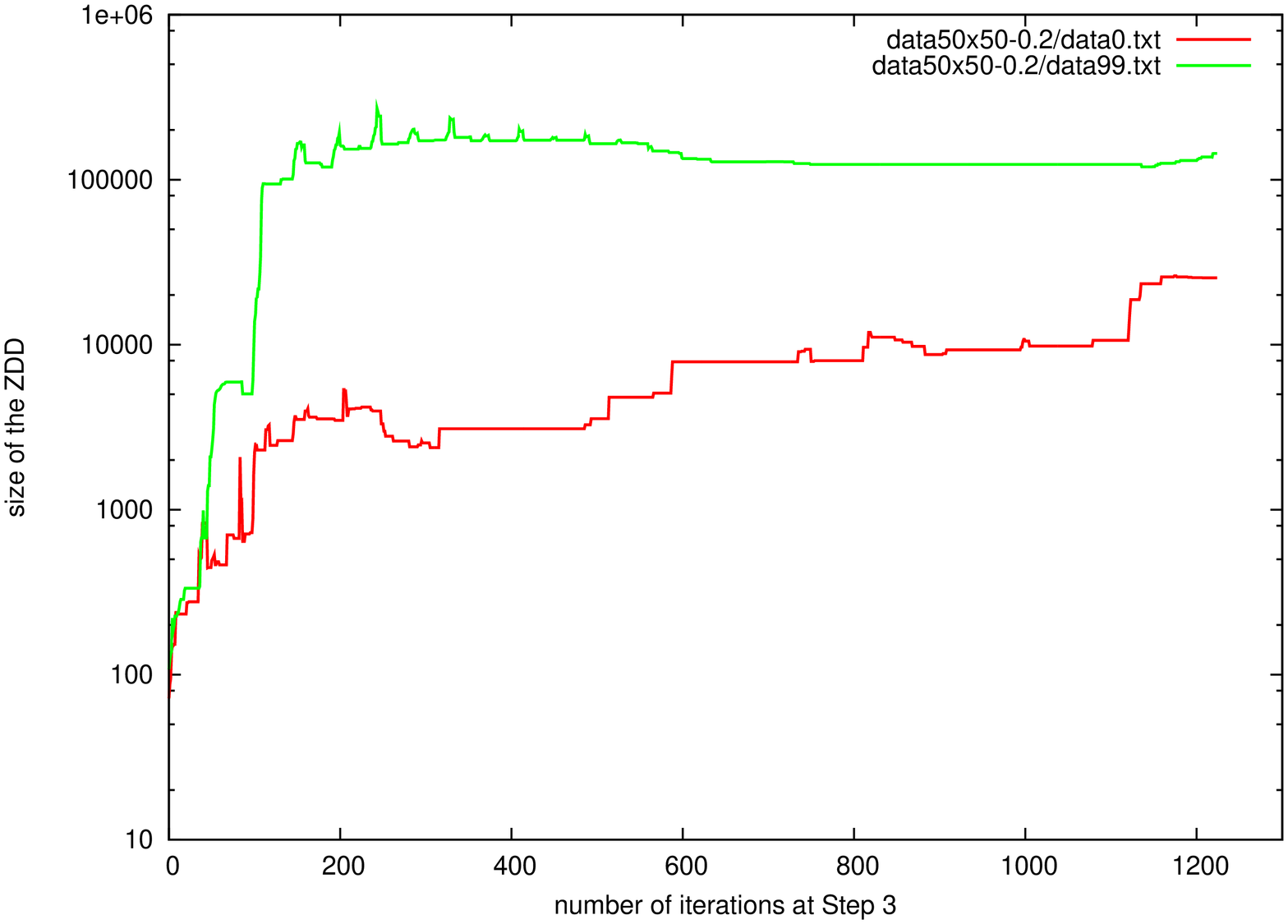}}}
\caption{The size of ZDDs during the execution of $\mathsf{ZDD}$\@.}
\label{fig:intermediate}
\end{figure}

\section{Conclusion}
\label{sec:conclusion}

We have presented the algorithm $\mathsf{ZDD}$ to 
enumerate all directed binary perfect phylogenies
from incomplete data, and compare it with the algorithm $\mathsf{B{\&}B}$ 
based on a simple branch-and-bound idea.
Theoretically, $\mathsf{B{\&}B}$ runs in polynomial time, but $\mathsf{ZDD}$  has 
no such guarantee.
In experiments, $\mathsf{ZDD}$ solved more instances than $\mathsf{B{\&}B}$\@.
This shows some gap between theory and practice, and it is desirable to have 
some theoretical justification why $\mathsf{ZDD}$ can outperform.
We have theoretically exhibited an example for which the compression by a ZDD is effective.
However, that example was artificial.
The experiments also show ZDD can compress very well on random instances.
It is desirable to obtain a more natural theoretical evidence why such a good compression is 
achieved.

The approach by ZDDs looks quite promising, and there must be more problems 
in bioinformatics that can get benefits from them.

\section*{Acknowledgments}
We thank Jesper Jansson for bringing the problem into our attention,
and Jun Kawahara and Yusuke Kobayashi for a fruitful discussion.
We also thank the anonymous referees of SEA 2012 for detailed comments.

\bibliographystyle{abbrv}
\bibliography{haplotype}


\appendix

\section{Appendix: Details for the Branch-and-Bound Enumeration Algorithm}
\label{sec:bb}

In our branch-and-bound algorithm, at every node of a search tree,
we make a decision whether a specified element $e$ of $S$ is contained in $S_j$ for 
a specified index $j$.
The following observation is easy to obtain.
\begin{lemma}
\label{lemma:recursion}
Let $S$ be a finite set,
$\cL = (L_1,\ldots,L_m)$ and
$\cU = (U_1,\ldots,U_m)$ be sequences of $m$ subsets of $S$ such that 
$L_i \subseteq U_i \subseteq S$ for all $i \in \{1,\ldots,m\}$,
and 
$\cS = (S_1,\ldots,S_m)$ be a directed binary perfect phylogeny for 
$\cL$ and $\cU$.
\begin{enumerate}
\item If $L_i = U_i$ for all $i\in\{1,\ldots,m\}$, then 
$\cS$ is a unique directed binary perfect phylogeny for $\cL$ and $\cU$.
\item If $e\in S_j \setminus L_j$ for some $j$, then $\cS$ is a directed binary
perfect phylogeny for $\cL'$ and $\cU$, where 
$\cL' = (L'_1, \ldots, L'_m)$ is defined as $L'_i = L_i$ for $i\neq j$ and
$L'_j = L_j \cup \{e\}$.
\item If $e\in U_j \setminus S_j$, for some $j$, then $\cS$ is a directed binary
perfect phylogeny for $\cL$ and $\cU'$, where $\cU'=(U'_1,\ldots,U'_m)$ is defined as
$U'_i=U_i$ for $i\neq j$ and $U'_j = U_j \setminus \{e\}$.
\qed
\end{enumerate}
\end{lemma}

Lemma~\ref{lemma:recursion} suggests the following algorithm.
Step 1 is the bounding step, and Step 3 is the branching step.

\begin{description}
\item[Algorithm:] $\mathsf{B{\&}B}(S,\cL,\cU)$
\item[Precondition: ] $S$ is a finite set, $\cL=(L_1,\ldots,L_m)$, $\cU=(U_1,\ldots,U_m)$, each
member of $\cL$ and $\cU$ is a subset of $S$, and 
$L_i \subseteq U_i$ for every $i\in\{1,\ldots,m\}$.
\item[Postcondition:  ] Output all the directed binary perfect phylogenies for $(S,\cL,\cU)$.
\item[Step 1: ] If there exists no directed binary perfect phylogeny
for $\cL$ and $\cU$, then output nothing and halt.
\item[Step 2: ] Otherwise, if $L_i=U_i$ for all $i\in\{1,\ldots,m\}$, then set $S_i = L_i$ for all
$i\in\{1,\ldots,m\}$, output $(S_1,\ldots,S_m)$ and halt.
\item[Step 3:] Otherwise, let $j\in\{1,\ldots,m\}$ be an arbitrary index such that 
$L_j \neq U_j$.
Choose an arbitrary element $e\in U_j \setminus L_j$.
\item[Step 3-1: ] Let $\cL' := (L'_1,\ldots,L'_m)$ be defined as $L'_i = L_i$ for all $i\neq j$,
and $L'_j=L_j\cup\{e\}$.  
Then, run  $\mathsf{B{\&}B}(S,\cL',\cU)$.
\item[Step 3-2: ] Let $\cU' := (U'_1,\ldots,U'_m)$ be defined as $U'_i = L_i$ for all $i\neq j$,
and $U'_j=U_j\setminus\{e\}$. 
Then, run  $\mathsf{B{\&}B}(S,\cL,\cU')$.
\item[Step 4:] Halt.
\end{description}

At Step 1, we may use any algorithm to check whether an instance $(S,\cL,\cU)$ admits a 
directed binary perfect phylogeny, e.g.\ one by Pe'er et al.\ \cite{DBLP:journals/siamcomp/PeerPSS04}.
Their algorithm actually outputs a directed binary perfect phylogeny $\cS=(S_1,\ldots,S_m)$ for
$(S,\cL,\cU)$ if it exists.
This $\cS$ can be used as further information, for example at Step 3 of Algorithm $\mathsf{B{\&}B}$.
We choose $e \in U_j \setminus L_j$ there.
We have two cases.
Remind that $L_j \subseteq S_j \subseteq U_j$ (by definition) and 
$L_j \neq U_j$ (by Step 2).
\begin{enumerate}
\item If $e\in S_j \setminus L_j$, then in the call $\mathsf{B{\&}B}(S,\cL',\cU)$
at Step 3-1 we do not have to perform Step 1
since $\cS$ is a directed binary perfect phylogeny for $(S,\cL',\cU)$.
\item If $e\in U_j \setminus S_j$, then in the call $\mathsf{B{\&}B}(S,\cL,\cU')$
at Step 3-2 we do not have to perform Step 1
since $\cS$ is a directed binary perfect phylogeny for $(S,\cL,\cU')$.
\end{enumerate}

The correctness of the algorithm is immediate.
We now bound the running time.
The relevant parameters are $m$, $n=|S|$, $k=\sum_{i=1}^{m}|U_i \setminus L_i|$, and
the number $h$ of output directed binary perfect phylogenies.
Let $t(m,n,k)$ be the worst-case time complexity of the algorithm that we use for Step 1.
Also, let $T(m,n,k,h)$ be the worst-case time complexity of 
the execution of $\mathsf{B{\&}B}(S,\cL,\cU)$ with these parameters.
If $k=0$, then $T(m,n,k,h) = O(mn)$
since Step 2 already takes $O(mn)$ time.
If $h=0$, then $T(m,n,k,h) = O(mn) + t(m,n,k)$.
Otherwise, 
\begin{displaymath}
T(m,n,k,h) \leq T(m,n,k-1,h_1)+T(m,n,k-1,h_2)+ O(mn) + t(m,n,k),
\end{displaymath}
where $h=h_1+h_2$.
This leads to
$T(m,n,k,h) \leq O(kh(mn + t(m,n,k)))$.

If we use the algorithm by 
Pe'er et al.\ \cite{DBLP:journals/siamcomp/PeerPSS04}, which runs in $\tilde{O}(mn)$ time,\footnote{%
The $\tilde{O}$-notation suppresses the polylogarithmic factor.
} at Step 1,
then 
we obtain the following theorem.

\begin{theorem}
The execution $\mathsf{B{\&}B}(S,\cL,\cU)$ correctly outputs all the directed binary
perfect phylogenies for $(S,\cL,\cU)$ without duplication in time
$\tilde{O}(mnkh)$ time, where $m$ is the length of the sequences $\cL,\cU$, 
$n=|S|$, $k=\sum_{i=1}^{m}|U_i \setminus L_i|$, and
the number $h$ of output directed binary perfect phylogenies.
In particular, each directed binary perfect phylogeny can be found in 
polynomial time (in the input size) per output, in the amortized sense.
\qed
\end{theorem}

For the experiment in Section \ref{sec:experiment}, 
we use the deterministic version of Algorithm A in the paper by Pe'er et al.\ 
\cite[p.\ 598]{DBLP:journals/siamcomp/PeerPSS04} as a subroutine
in Step 1, but we have simplified it to gain a 
practical performance.

\end{document}